\documentclass{article}

\usepackage[margin=1.5in]{geometry}

\usepackage{amsthm}
\usepackage{amsfonts}
\usepackage{amsmath}
\usepackage{graphicx}
\usepackage{url}
\usepackage{hyperref}

\newtheorem{theorem}{Theorem}
\newtheorem{lemma}[theorem]{Lemma}
\newtheorem{proposition}[theorem]{Proposition}
\newtheorem{corollary}[theorem]{Corollary}
\newcommand\term[1]{{\it #1}}

\newcommand\R{\mathbb R}

\newcommand\bdy\partial

\newcommand\comment[1]{}

{\makeatletter
 \gdef\xxxmark{%
   \expandafter\ifx\csname @mpargs\endcsname\relax 
     \expandafter\ifx\csname @captype\endcsname\relax 
       \marginpar{xxx}
     \else
       xxx 
     \fi
   \else
     xxx 
   \fi}
 \gdef\xxx{\@ifnextchar[\xxx@lab\xxx@nolab}
 \long\gdef\xxx@lab[#1]#2{\textbf{[\xxxmark #2 ---{\sc #1}]}}
 \long\gdef\xxx@nolab#1{\textbf{[\xxxmark #1]}}
 \long\gdef\xxx@lab[#1]#2{}\long\gdef\xxx@nolab#1{}%
}

\begin{document}

\title{Generalized D-Forms Have No Spurious Creases}
\author{Erik~D.~Demaine\footnote{Partially supported by NSF CAREER
    award CCF-0347776, DOE grant DE-FG02-04ER25647, and AFOSR grant
    FA9550-07-1-0538.}$\;$
 and Gregory~N.~Price\footnote{Partially supported by an NSF
    Graduate Fellowship.}\\
  MIT Computer Science and Artificial Intelligence Laboratory \\
  32 Vassar St., Cambridge, MA 02139, USA \\
  \protect\url{{edemaine,price}@mit.edu}
} \date{}
\maketitle

\begin{abstract}
  A convex surface that is flat everywhere but on finitely many
  smooth curves (or \emph{seams}) and points is a \emph{seam form}.
  We show that the only creases through the flat components of a seam
  form are either between vertices or tangent to the seams.  As
  corollaries we resolve open problems about certain special seam
  forms: the flat components of a D-form have no creases at all, and the
  flat component of a pita-form has at most one crease, between the seam's
  endpoints.
\end{abstract}

\newpage

\section{Introduction} \label{sec:introduction}

Given any metric space $S$ with the topology and local geometry
required of the surface of a convex three-dimensional body,
there is exactly one convex body up to isometry whose surface has
the intrinsic geometry of $S$.  This is the endpoint of a line of
research pursued in the middle of the last century by Alexandrov and
Pogorelov \cite{Alexandrov}, and it implies a strong correspondence
between the geometry of a convex body and the intrinsic geometry of
its surface.  On the other hand, the exact nature of this
correspondence is not yet well understood---which properties in the
surface geometry imply what properties in the body geometry, and vice
versa.

For example, if $S$ is a \emph{D-form}, obtained by sewing together
two smooth convex shapes of the same perimeter, then even for this
special case the most basic questions are open.  These forms were
invented by an artist \cite{Wills} and introduced into the literature
in \cite{PW}.  The latter study poses three
problems: (1) when is the D-form the convex hull of a space curve, (2)
when are the two pieces free of creases, and (3) how can one compute the
D-form numerically from the two shapes.  A later treatment
\cite{GFALOP} suggests an informal argument for Problem 1 (arguing
that the D-form is always the convex hull of its seam) and leaves
Problems 2 and 3 open.  The same book \cite{GFALOP} introduces also a related
special case where $S$ is obtained by sewing up a single smooth convex
shape along its boundary in one seam, calling these \emph{pita-forms}
and suggesting, based on paper experiments, that pita-forms might never
have creases.

We resolve Problems 1 and 2: both D-forms and pita-forms are always
the convex hull of their seams, and (excluding the seam) D-forms are
always free of creases but a pita-form may have one crease.  Our
results apply to a natural generalization of both D-forms and
pita-forms, the \emph{seam form}, which
roughly consists of intrinsically flat pieces joined along finitely many
seams.  Because the original sources of the problem are stated
informally, we first introduce precise definitions that we believe
capture the intuitive picture.  Then we show the following theorems:

\begin{theorem}\label{t:convex-hull-general}
  Every three-dimensional convex body is the convex hull of the
  nonflat points on its surface.
\end{theorem}

\begin{corollary}\label{c:convex-hull-seam-form}
  Every seam form is the convex hull of its seams and vertices.
\end{corollary}

\begin{theorem}\label{t:no-creases-seam-form}
  In a flat component of a seam form, every crease lies on a line
  segment composed of creases, and each endpoint of such a segment is
  either a strict vertex or a point of tangency to a seam.
\end{theorem}

\begin{corollary}\label{c:no-creases-d-form}
  The flat components of a D-form are without creases; in the flat
  component of a pita-form, the only crease(s) make up the line
  segment between the endpoints of the seam.
\end{corollary}

Intuitively the line segment between a pita-form's endpoints in
Corollary~\ref{c:no-creases-d-form} should be thought of as one
``crease''; it is a consequence of our definitions, below, that this
segment may be arbitrarily subdivided into several segments we call
creases.

Problem 3, to efficiently compute the three-dimensional shape of a
D-form or seam form from its two-dimensional intrinsic geometry, has
now been largely resolved.  To make this problem well posed, one needs
a finite representation of the input geometry, which is most naturally
done by a piecewise-linear or polyhedral approximation.  
With considerable effort, the
problem of reconstructing a three-dimensional convex polyhedron from
its intrinsic geometry can be
reduced to the solution of a high-dimensional ordinary differential
equation~\cite{BI}.  The numerical solution of this equation appears
to be achievable efficiently in practice, and is provably achievable
within pseudopolynomial time~\cite{KPD}.

We introduce terminology in Section~2, prove
Theorem~\ref{t:convex-hull-general} and its
Corollary~\ref{c:convex-hull-seam-form} in Section~3, and prove
Theorem~\ref{t:no-creases-seam-form} and its
Corollary~\ref{c:no-creases-d-form} in Section~4.  In Section~5 we
describe counterexamples that show the necessity of some of the
hypotheses in our results.

\section{Background and Notation} \label{sec:background}

For us a \term{surface} is a metric 2-manifold embedded in $\R^3$.  The
surface is $C^k$ if the manifold and its embedding are $C^k$.  The
surface is \term{piecewise-$C^k$} if it can be decomposed as a complex
of vertices, $C^k$ open edges, and $C^k$ open regions.

A \term{good surface} is a piecewise-$C^2$ surface.  A good surface
$S$ therefore decomposes into a union of $C^2$ surfaces $S_i$, called
\term{pieces}, $C^2$ edges $\gamma_j$, which we call
\term{semicreases}, and vertices.  If $S$ is itself $C^1$ everywhere on a
semicrease, we call it a \term{proper semicrease}; otherwise it is a
\term{crease}.  (This conservative definition of crease,
where some parts may be $C^1$ but not~$C^2$,
only broadens our characterization of creases in seam forms.)

A point on a surface is \term{flat} if it has a neighborhood isometric
to a region in the plane.  A surface or part of a surface is called
flat if all of its points are flat.

A surface $S$ is \term{convex} if $S \subseteq \bdy X$ for some bounded
convex body $X$ in $\R^3$.
A \term{normal} to a convex
body $X$ at a point $x$ is a unit vector $n$ with $n \cdot x =
\sup_{x' \in X} n\cdot x'$.
The relation between points on the boundary of $X$ and their normals
is traditionally called the \term{Gauss map}, though it need not be a
map---one point may have many normals.  We write $G(x)$ for the
normals at $x$, and $G(U)$ for all the normals to any point in $U
\subset X$.  Observe that $G(x)$ is always a convex subset of the sphere.

A consequence of Gauss' celebrated Theorema Egregium~\cite{Gauss}
is that a convex surface $U$ is flat just if $G(U)$ has zero area.
If $G(x)$ has positive spherical area, then we call $x$ a \term{strict
  vertex}.  The $C^2$ condition prevents a strict vertex $x$ from being
on a semicrease or a piece, so for good surfaces, strict vertices
are indeed vertices.

\begin{figure}[tb]
  \centering
  \includegraphics[scale=.4]{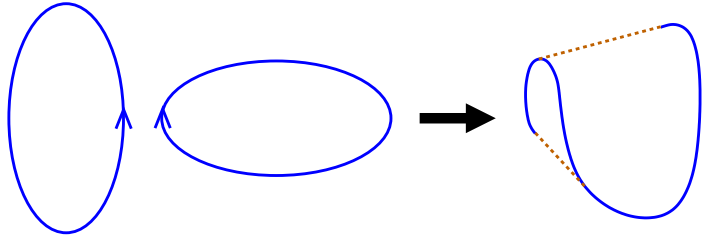}
  \caption{A D-form, constructed by sewing together two ellipses (at
    left).  The solid line is the seam, and the dotted lines are
    the false edges from projection into the page.}
  \label{f:d-form}
\end{figure}

A \term{seam form} $S$ is a good convex surface in which each piece
$S_i$ is flat.  It is simple to verify that such a surface decomposes
uniquely into maximal connected flat open subcomplexes, which we call
\term{flat components}, and some leftover semicreases and vertices,
the connected components of which we call \term{seams}.

A \term{simple seam form} is a seam form in which each flat component
is isometric to a convex plane region, and a \term{simple smooth seam
  form} is a simple seam form in which these regions have smooth
($C^\infty$) boundary.
A simple smooth seam form with one flat component is called a \term{pita
  form}, and with two flat components is called a \term{D-form}.
See Figures \ref{f:d-form} and \ref{f:pita-form}.

Given a convex body $X$ and a point $x
\in X$, we say that a line $\ell$ \term{traverses} $x$ if $x \in \ell$
and some open neighborhood of $x$ on $\ell$ is contained in~$X$.  An
\term{extreme point} of $X$ is a point $x \in X$ not traversed by any
line.

\section{Convex Hull}
\label{sec:convex-hull}

In this section we prove Theorem~\ref{t:convex-hull-general}, that
every three-dimensional convex body is the convex hull of the nonflat
points on its surface.

From convex geometry we have the following characterization of the
minimal set from which a convex body can be recovered as the convex
hull: we need only the extreme points.
\begin{theorem}[Minkowski's Theorem]\label{t:hull-of-extreme}
  Every convex body in $\R^n$ is the convex hull of its extreme points.
\end{theorem}
The proof is a straightforward induction on the dimension of the body
and can be found as Corollary 1.4.5 in Schneider's textbook
\cite{Schneider}.

It remains to describe the extreme points of a seam form.  To do so we
begin with the following proposition:
\begin{proposition}\label{p:extreme-separate}
  If $p$ is an extreme point of the convex body $X$, then for every
  open neighborhood $U$ of $p$ in $X$, some hyperplane has $p$
  strictly on one side and all of $X \setminus U$ strictly on the other.
\end{proposition}
\begin{proof}
  This is Lemma 1.4.6 in \cite{Schneider}.  For completeness we
  give the proof.

  Because $p$ is an extreme point, it cannot by definition be the
  convex combination of any two other points in $X$.  Therefore it is
  not the convex combination $a_1x_1 + \dotsb + a_kx_k$, with all $a_i
  > 0$, of any $k$ other points in $X$, because otherwise we would
  have $p = (1-a_k)(\frac{a_1}{1-a_k}x_1 + \dotsb +
  \frac{a_{k-1}}{1-a_k}x_{k-1}) + a_kx_k$, and certainly the convex
  combination $\frac{a_1}{1-a_k}x_1 + \dotsb +
  \frac{a_{k-1}}{1-a_k}x_{k-1}$ is a point in $X$.  In other words,
  $p$ lies outside the convex hull of $X \setminus \{p\}$, and
  consequently outside the convex hull $Y$ of $X \setminus U$.

  Now $Y$ is itself a convex body, and $p$ a point outside it.  By the
  Separating Hyperplane Theorem, some hyperplane strictly separates
  them, and because $X \setminus U \subset Y$, it strictly separates $p$
  and $X \setminus U$ as required.
\end{proof}

\begin{proposition}\label{p:flat-no-extreme}
  On the surface of a convex body, there are no flat extreme points.
\end{proposition}
\begin{proof}
  Suppose some extreme point $p$ of a convex body $X$ was flat,
  with a neighborhood $S \subset \partial X$ isometric to a plane region.
  Let $U$ be an open neighborhood of $p$ in $X$ with $U \cap \partial
  X \subset S$.  Let the hyperplane $H$ guaranteed by
  Proposition~\ref{p:extreme-separate} separate $X$ into convex bodies
  $C$ and $Y$ with $p \in C$, and let $D = C \cap \partial X$.  Because $C
  \subset U$, we have $D \subset S$ so that $D$ is flat.

  Now consider the normals to $X$ along the portion $D$ of its
  surface.  Let $d$ be the distance from $p$ to $H$, and let $r$ be
  the maximum distance from the projection of $p$ onto $H$ to any
  point in $H \cap X$.  Then any plane through $p$ and making an angle
  at most $\theta = \tan^{-1}(d/r)$ to $H$ fails to intersect $H \cap X$ and
  therefore fails to intersect $Y$.  Therefore the normals to these
  planes, covering a spherical area of $2\pi(1 - \cos \theta) > 0$,
  all are normals to $X$ somewhere on~$D$.  This gives $G(D)$
  a positive area, contradicting that $D \subset S$ is flat.
\end{proof}

Theorem~\ref{t:convex-hull-general} is now immediate from
Theorem~\ref{t:hull-of-extreme} and Proposition~\ref{p:flat-no-extreme},
and Corollary~\ref{c:convex-hull-seam-form} follows.

\section{Creases}
\label{sec:creases}

In this section we prove Theorem \ref{t:no-creases-seam-form},
characterizing the possible creases of a seam form.

\comment{
\begin{lemma}\label{l:gauss-map-closed}
  For a convex body $X$, the Gauss ``map'' $G$, regarded as a relation
  in $X \times S^2$, is a closed set.
\end{lemma}
\begin{proof}
  We need to show that if points $x_1, x_2, \dotsc \in X$ converge to
  $x_* \in X$, and normals $n_1, n_2, \dotsc \in S^2$ to $x_1, x_2,
  \dotsc$ respectively converge to $n_* \in S^2$, then $n_*$ is a
  normal to $x_*$.  We have $n_i \cdot x_i = \sup_{x' \in X} n_i \cdot
  x'$ for each $i$.  By continuity, $n_i \cdot x_i$ converges to $n_*
  \cdot x_*$.  Because $X$ is bounded, the functions $x' \mapsto n_i
  \cdot x'$ converge uniformly to $x' \mapsto n_* \cdot x'$, and the
  right-hand sides $\sum_{x' \in X} n_i \cdot x'$ converge to
  $\sup_{x' \in X} n_* \cdot x'$.  Consequently $n_* \cdot x_* =
  \sup_{x' \in X} n_* \cdot x'$ and $n_* \in G(x_*)$ as required.
\end{proof}
}

\begin{proposition}\label{p:flat-crease-segment}
  Let $\gamma$ be a crease in a flat component of a seam form $S$.
  Then $\gamma$ lies on a line segment $[p,q]$ between endpoints $p$ and $q$
  that lie on seams, and the whole segment is composed of creases and vertices.
\end{proposition}
\begin{proof}
  Let $S_1$ and $S_2$ be the open pieces bordered by $\gamma$
  in the decomposition of the good surface~$S$,
  and let $x \in \gamma$ be a point at which $S$ is not~$C^1$.
  Then $S_1$ and $S_2$ are $C^2$ surfaces,
  so they have normals $n_1$ and $n_2$ at $x$,
  and because $S$ is not $C^1$ at $x$, these normals are distinct.
  Therefore $G(x)$ contains at least two distinct vectors.

  By Proposition~\ref{p:flat-no-extreme}, $x$ must be traversed by
  some line $\ell$, so that $\ell \cap S = [p, q]$ for some $p$ and $q$.
  Necessarily $\ell$ is perpendicular to all of $G(x)$,
  so for each $y \in [p, q]$ and each normal $n \in G(x)$,
  $n \cdot y = n \cdot x = \sup_{x' \in X} n \cdot x'$ and $n$ is a
  normal of $y$.  Therefore each $G(y)$ contains $G(x)$, and
  so like $G(x)$ has at least
  two distinct vectors.

  The multiple normals in $G(x)$ and hence in each $G(y)$ determine a unique perpendicular line, so
  that no other line may traverse any point of $[p, q]$.  In
  particular no line traverses $p$ or $q$, so by
  Proposition~\ref{p:flat-no-extreme}, these points are not flat and
  must lie on seams or vertices.

  At the same time, because a $C^1$ surface has only one normal at
  each point, no point of $[p, q]$ can be on a $C^2$ piece or a
  semicrease.  The whole segment is therefore made up of creases and
  (nonstrict) vertices.  Because a crease is defined from a
  cell-complex decomposition, only one crease runs through a given
  point, so because $\gamma$ runs through $x$ it must be one of the
  creases making up $[p, q]$.
\end{proof}

In order to analyze the Gauss map at seam and vertex points, we
introduce some additional notation.  Let $x \in S$ be incident to a
1- or 2-cell $C$, a (semi)crease or piece.  Then we define
$$ G_C(x) = \varliminf_{U \ni x} \overline{G(C \cap U)} $$
as the \term{Gauss map at $x$ on $C$}.  For comparison,
observe that $G(x) = \varliminf_{U \ni x} \overline{G(U)}$ because the relation
$G$ is closed, and in particular $G_C(x) \subseteq G(x)$.

\begin{proof}[Proof of Theorem \ref{t:no-creases-seam-form}]
  Let $\gamma$ be a crease in a flat component of a seam form $S$.  By
  Proposition~\ref{p:flat-crease-segment}, $\gamma$ lies on a segment
  $[p, q]$ composed of creases and vertices and whose endpoints lie on seams.  It remains
  to prove that if an endpoint, say $p$, lies on a seam and is not a
  strict vertex, then the seam is tangent to $[p, q]$.

  Let $G_\gamma(p)$ be the great circular arc $mn$, and let the pieces of $S$
  bordering $\gamma$ be $S_1$ and $S_2$.  By continuity,
  $G_{S_1}(p)
  \ni m$ and $G_{S_2}(p) \ni n$ (possibly after exchanging the names
  $m, n$), and because $m \neq n$ the Gauss map at $p$ on at least one
  of the cells $C$ surrounding $p$ from $S_1$ to $S_2$ apart from
  $\gamma$ must be a positive-length spherical curve in order to
  complete the path from $m$ to $n$.
  If $p$ is not a strict vertex, then $G(p)$ is a convex spherical
  shape of zero area, so it is a great circular arc,
  and $G_C(p) \subseteq G(p)$ is also a great circular arc.
  If $C$ is a piece, then $G_C(p)$ is either a singleton or a
  curve not lying on a great circle, because a great-circle Gauss map
  makes parallel rule lines that cannot converge at $p$.
  Therefore $C$ is a semicrease.  Because $G(C)$ must be more than a single
  point, $C$ is a crease, and to make the Gauss map lie within
  the arc $G(p)$, $C$ must be tangent to $[p,q]$ as required.  Finally,
  because $p$ is the endpoint of the intersection of the line $pq$
  with $S$, the crease $C$ must not be a line segment, so by
  Proposition~\ref{p:flat-crease-segment},
  it is actually part of the seam and the proof is complete.
\end{proof}

\comment{
\begin{lemma}\label{l:no-pc-one-crease}
  If a point $p$ on the surface of a convex body is not a vertex, then either
  \begin{enumerate}
  \item no creases pass through $p$, nor approach with nonzero
    limiting dihedral angle at $p$; or
  \item all creases with nonzero limiting dihedral angle at $p$ are
    tangent to a common line
  \end{enumerate}
\end{lemma}
\begin{proof}
  The image $G(p)$ of the Gauss map must be a convex set of zero area,
  because $p$ is not a vertex.  The only such figures on the sphere are
  the singleton points and the great circular arcs.

  If $G(p)$ is a singleton point, then $p$ has only a single
  normal and tangent plane; so no crease can have nonzero dihedral
  angle at $p$, as it would bring two distinct normals and a
  great circular arc between them.

  If $G(p)$ is a great circular arc, then it corresponds to a dihedral
  angle about the line perpendicular to this arc; any nonzero crease
  at $p$ not tangent to this line would introduce an arc of normals
  outside of $G(p)$, so no such crease can be present.
\end{proof}
}

Of course, in a convex plane region, no line segment in the interior is
tangent to the boundary, from which follows a corollary about simple
seam forms.

\begin{corollary}\label{c:no-creases-doubly-convex}
  In a simple seam form, every crease in a flat component
  is on a line segment between two strict vertices.
\end{corollary}

Finally, in a simple smooth seam form such as a pita-form or a D-form,
the requirement of smoothness sharply limits the possible
configurations.  By (local) convexity, no vertex can be incident to three or
more semicreases as part of its seam, and a vertex through which a seam
passes cannot be a strict vertex.  Consequently a pita-form must have
a single path for its seam and just two strict vertices located at the
seam's endpoints, and a D-form must have a single cycle for its seam
and no strict vertices.  Corollary \ref{c:no-creases-d-form} follows.

\section{Counterexamples}
\label{sec:counterexamples}

We have required the flat components of a D-form to be convex.  We
could relax this
requirement, requiring instead only that the metric space resulting
from joining the two components be locally convex, and
the~Alexandrov-Pogorelov theorem would still guarantee a unique convex
embedding in three-dimensional space.  Of course Corollary
\ref{c:convex-hull-seam-form} would still guarantee that the resulting
body would be the convex hull of its seam, but it turns out that
Corollary \ref{c:no-creases-d-form}, whose conditions would no longer
be satisfied, really would fail in its conclusion: one can construct a
``D-form'' under this relaxed definition which contains creases in its
flat components.  Indeed it is not hard to construct such an example, if
one keeps in mind Theorem \ref{t:no-creases-seam-form} that the
offending crease must be tangent to a seam; see Figure
\ref{f:d-form-creased}.

\begin{figure}[tbh]
  \centering
  \includegraphics[scale=.3]{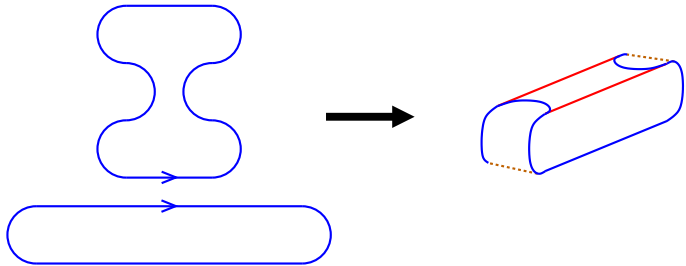}
  \caption{A ``D-form'' with a relaxed convexity condition.  The 
    solid dark line is the seam, the dotted lines are false edges from
    projection onto the page, and the solid light lines are the
    creases through a flat component.  Actual D-forms have no such
    creases.}
  \label{f:d-form-creased}
\end{figure}

\xxx{Should obviate need for color, esp. in captions, for final version.}

For pita-forms, we have concluded in Corollary
\ref{c:no-creases-d-form} that a pita-form may have at most one
crease.  Indeed this is tight, and it is easy to construct an example
pita-form with a crease; see Figure \ref{f:pita-form}.  This
possibility of creases therefore represents a real difference from
D-forms.  It represents also a contrast from the appearance of the
natural paper experiments, which led the authors first introducting
pita-forms (\cite{GFALOP}) to suggest that pita-forms might never have
creases; in fact, once one is familiar with Corollary
\ref{c:convex-hull-seam-form}, it is clear that the same experiments
really would have to be creased if~only the paper were behaving
ideally.

\begin{figure}[tbh]
  \centering
  \includegraphics[scale=.5]{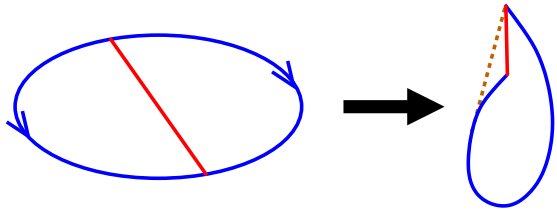}
  \caption{A typical pita-form.  A crease runs between the endpoints of the seam.}
  \label{f:pita-form}
\end{figure}

For some theorems of the same flavor as our results here, one might
hope to obtain proofs by showing that the desired properties hold of
convex polyhedra, which are relatively concrete and amenable to
reasoning, and then that they carry over to general convex bodies as
limits of polyhedra.  Indeed, this is the approach taken in
\cite{GFALOP} to argue for Corollary \ref{c:convex-hull-seam-form} for
the case of D-forms.
Unfortunately this approach does not hold as widely as one might like.
In particular, Corollary \ref{c:no-creases-doubly-convex} cannot be
proven by a limiting argument of the obvious form, even when
restricted to D-forms.  When each
flat component of the D-form is approximated by a sequence of
polygons, it is possible for the dihedral angles inside the
components to approach positive limits, even when the angles of the polygonal
approximations are required to converge to zero.  For example, in
the D-form obtained from two circular disks (which is just the double
cover of a disk), the components may be approximated by regular $n$-gons
for increasing $n$, and the resulting approximations to the D-form may
be antiprisms of two smaller $n$-gons and $2n$ triangles.  In this
approximation sequence, the dihedral angles between each $n$-gon and
its $n$ neighboring triangles approach $\pi/4$, not zero, even though
they lie inside the flat components.  For this and other reasons we
have chosen direct proofs that attack the general case of convex bodies.

\section{Acknowledgments}
We thank Jonathan Kelner, Joseph O'Rourke, and Johannes Wallner for
helpful discussions, and the anonymous referees for helpful comments.

\bibliographystyle{plain}

\begin{thebibliography}{PW01}
\bibitem[Ale50]{Alexandrov}A. D. Alexandrov,  {\sl Convex Polyhedra},
  Springer-Verlag, Berlin, 2005.  See especially note 21, page 189.
\bibitem[BI06]{BI}Alexander I. Bobenko and Ivan~Izmestiev, Alexandrov's
  theorem, weighted Delaunay triangulations, and mixed volumes,
  Annales de l'Institut Fourier, in press. {\rm arXiv:math.DG/0609447}.
\bibitem[DO07]{GFALOP}Erik D. Demaine and Joseph O'Rourke, {\sl
    Geometric Folding Algorithms}, Cambridge University Press,
  Cambridge, 2007.  Pages 352--354.
\bibitem[Gau02]{Gauss}Gauss, Carl Friedrich, General Investigations of
  Curved Surfaces, 1827, Morehead and Hiltebeitel, tr., Princeton, 1902.
\bibitem[KPD09]{KPD}Daniel Kane, Gregory N. Price, and Erik
  D. Demaine, A pseudopolynomial algorithm for Alexandrov's theorem,
  Algorithms and Data Structures Symposium (WADS) 2009.  {\rm arXiv:0812.5030}.
\bibitem[PW01]{PW}Helmut Pottmann and Johannes Wallner, {\sl
    Computational Line Geometry}, Springer-Verlag, Berlin, 2001.  Page
  418.
\bibitem[Sch93]{Schneider}Rolf Schneider, {\sl Convex Bodies: The
    Brunn-Minkowski Theory}, Cambridge University Press, Cambridge,
  1993.
\bibitem[Wil]{Wills}Tony Wills, DForms: 3D forms from two 2D sheets,
  {\sl Bridges: Mathematical Connections in Art, Music, and Science},
  Reza~Sarhangi and John~Sharp, eds, London, pp.~503--510.
\end{thebibliography}

\end{document}